\documentclass{article}

\usepackage[utf8]{inputenc}
\usepackage[left=3cm, right=3cm, top=2cm, bottom=2cm]{geometry}

\usepackage[linesnumbered,boxed]{algorithm2e}
\usepackage[noend]{algpseudocode}
\usepackage{mdframed}

\usepackage{amsmath}
\usepackage{amssymb}
\usepackage{amsthm}
\usepackage{authblk}
\usepackage{comment}
\usepackage{xcolor}
\usepackage{tikz}
\usepackage[noadjust]{cite}

\newtheorem{definition}{Definition}
\newtheorem{theorem}{Theorem}

\newtheorem{lemma}[theorem]{Lemma}

\newcommand{\dist}{\text{dist}}

\newcommand{\Oish}{\widetilde{O}}

\title{A Note on Distance-Preserving Graph Sparsification\footnote{A previous version of this paper claimed a proof of a $+4$ all-pairs additive spanner on $O(n^{7/5})$ edges, in addition to the other results presented here.  With apologies, this result has been retracted due to a fatal bug in the argument.  As of this writing, state-of-the-art for the $+4$ spanner is $\Oish(n^{7/5})$ edges, by Chechik \cite{Chechik13soda}.}}
\author{Greg Bodwin\footnote{Research performed while the author was affiliated with Georgia Tech, and supported by NSF awards CCF-1717349, DMS-183932 and CCF-1909756.}\\ \texttt{bodwin@umich.edu}\\ University of Michigan EECS}
\date{}

\begin{document}



\maketitle

\thispagestyle{empty}

\begin{abstract}
We consider problems of the following type: given a graph $G$, how many edges are needed in the worst case for a sparse subgraph $H$ that approximately preserves distances between a given set of node pairs $P$?
Examples include pairwise spanners, distance preservers, reachability preservers, etc.
There has been a trend in the area of simple constructions based on the \emph{hitting set technique}, followed by somewhat more complicated constructions that improve over the bounds obtained from hitting sets by roughly a $\log$ factor.
In this note, we point out that the simpler constructions based on hitting sets don't actually need an extra $\log$ factor in the first place.
This simplifies and unifies a few proofs in the area, and it improves the size of the $+4$ pairwise spanner from $\Oish(np^{2/7})$ [Kavitha Th.\ Comp.\ Sys.\ '17] to $O(np^{2/7})$.
\end{abstract}

\pagebreak

\pagenumbering{arabic} 

\section{Introduction}

In graph algorithms, an effective preprocessing technique is to replace a large input graph with a ``similar'' smaller graph, which can thus be stored or analyzed more efficiently in place of the original.
In this paper, we will specifically study sparsification problems where the goal is to find a sparse subgraph that approximately preserves \emph{some shortest path distances} of the input.
Our focus will be on the following objects:
\begin{definition} [Sparsifier Variants \cite{PU89jacm, LS91, CE06, CGK13, AB18}] \label{def:allspans}
Given a (possibly directed/weighted) graph $G = (V, E)$ and a set of demand pairs $P \subseteq V \times V$, a subgraph $H$ is a \emph{$+k$ pairwise spanner} of $(G, P)$ if we have
$$\dist_H(s, t) \le \dist_G(s, t) + k \qquad \text{for all } (s, t) \in P.$$
We say that a particular demand pair $(s, t)$ is \emph{satisfied} by a subgraph $H$ when the above inequality holds for $s, t$.
When $k=0$, i.e.\ distances between demand pairs are preserved exactly, we say that $H$ is a \emph{distance preserver}.
When $k=\infty$, i.e.\ the only requirement is to preserve reachability between demand pairs, $H$ is called a \emph{reachability preserver}.
\end{definition}

We will use the general term \emph{sparsifier} to be deliberately ambiguous to which of these objects is in play, so that we may speak about all of them at once.
For applications of these various sparsifiers to algorithms, data structures, routing schemes, etc., we refer to the recent survey \cite{ABSHJKS19}.
The most common goal in this area is to bound the extremal tradeoff between the error budget $k$, the number of demand pairs $|P|$, the number of nodes in the input graph $n$, and the number of edges needed in the sparsifier $|E(H)|$ (ideally $|E(H)|$ is as small as possible).

We will discuss a variant on the problem called \emph{sparsifiers with slack}, in which only a constant fraction of the demand pairs need to be satisfied:
\begin{definition} [Sparsifiers with Slack]
Given a graph $G$ and demand pairs $P$, a subgraph $H$ is a \emph{sparsifier with slack} if there is $P' \subseteq P, |P'| = \Omega\left(|P|\right)$, such that $H$ is a sparsifier of $(G, P')$.
\end{definition}

See \cite{CDG06, Dinitz07, KRXY07} for some prior work on various sparsifiers with slack; this is also closely related to the ``for-each'' setting studied for spectral sparsifiers and related objects \cite{ACKQWZ16, CGPSSW18}.
We will say \emph{``complete''} sparsifier when we want to emphasize that we mean a standard sparsifier, with all demand pairs satisfied, rather than a sparsifier with slack.
This paper is driven by the following observation relating the two settings:
\begin{lemma} [Main Lemma] \label{lem:firstmain}
Let $a, b, c>0$ be absolute constants, let $G$ be an $n$-node input graph, and let $p^*$ be a parameter.
Additionally suppose:
\begin{itemize}
\item there is a complete sparsifier on $O(n^a)$ edges for any set of demand pairs of size $|P| \le p^*$, and

\item there is a sparsifier with slack on $O(n^b |P|^c)$ edges for any set of demand pairs of size $|P| \ge p^*$.
\end{itemize}
Then there is a complete sparsifier on $O(n^a + n^b |P|^c)$ edges.
\end{lemma}
\begin{proof}
Let $\alpha$ be an absolute constant such that the sparsifier with slack satisfies at least an $\alpha$ fraction of the given demand pairs.
While $|P| \ge p^*$, compute a sparsifier with slack on $O(n^b |P|^c)$ edges, remove the satisfied demand pairs from $P$, and then repeat on the remaining demand pairs.
Once $|P| \le p^*$, compute one final complete sparsifier on $O(n^a)$ edges, and then union all computed sparsifiers together.

Let $P_i$ denote the demand pairs remaining in the $i^{th}$ round (the initial set of demand pairs is $P_0$).
To bound the total size of the sparsifiers with slack $\{H_i\}$ computed in each round, we have:
\begin{align*}
\left| \bigcup_i E(H_i) \right| \le \sum_i |E(H_i)| &= \sum_i O\left(n^b |P_i|^c \right)\\
&\le \sum_i O\left(n^b \left(|P_0|(1-\alpha)^{i}\right)^c \right) \tag*{$\alpha$ fraction satisfied each round}\\
&= O\left(n^b |P_0|^c\right) \cdot \sum_i (1-\alpha)^{ic}\\
&= O\left( n^b |P_0|^c \right) \tag*{telescoping sum.}
\end{align*}
So we pay $O(n^b |P|^c)$ for the sparsifiers with slack, and $O(n^a)$ for the final complete sparsifier, completing the proof.
\end{proof}

The relevance of this lemma passes through the \emph{hitting set technique}, a common method in the area where one randomly selects nodes of the graph and uses the random choices to inform the construction somehow, arguing that (if the sample is large enough) then one hits every ``important'' part of the graph with high probability.
This generally leads to simple and elegant constructions, at the price of an extra $\log$ factor in the size needed to achieve the high probability guarantee.
Some research effort has been spent discovering somewhat more involved constructions that remove this $\log$ (see below).
The point of Lemma \ref{lem:firstmain} is that, since one only really needs a sparsifier with slack, it actually suffices for the hitting set to hit each important part of the graph with \emph{constant} probability.
This means the $\log$ factor can be removed directly from these hitting set arguments, allowing the simpler constructions to be used.

\begin{theorem} [Informal] \label{thm:sparsifiers}
The following theorems all have simple proofs based on the hitting set technique.
Let $n$ be the number of nodes in the input graph and $p$ the number of demand pairs.
\begin{itemize}
\item Every (possibly directed and weighted) graph has a distance preserver on $O(np^{1/2})$ edges. \cite{CE06}

\item Every (possibly directed) graph has a reachability preserver on $O((np)^{2/3} + n)$ edges. \cite{AB18}

\item Every undirected unweighted graph has a $+2$ pairwise spanner on $O(np^{1/3})$ edges. \cite{KV15, AB16soda, Knudsen14}

\item Every undirected unweighted graph has a $+4$ pairwise spanner on $O(np^{2/7})$ edges. \cite{Kavitha17}

\item Every undirected unweighted graph has a $+6$ pairwise spanner on $O(np^{1/4})$ edges. \cite{Kavitha17, Knudsen14}
\end{itemize}
\end{theorem}

Full proofs of these theorems, which are mostly just simplified and unified expositions of the corresponding constructions in prior work, are given in the body of the paper below.
The $+4$ pairwise spanner is actually a slightly improved result here; the previous bound was $\Oish(np^{2/7})$ \cite{Kavitha17}, as the $\log$ factors from the hitting set technique had not been previously shaved.

\section{Pairwise Sparsifiers with Slack}

We now begin proving the pieces of Theorem \ref{thm:sparsifiers}.
We start with the following foundational result in distance preservers:
\begin{theorem} [\cite{CE06}] \label{thm:cedp}
Any $n$-node directed weighted graph $G$ and $|P|=p$ demand pairs have a distance preserver on $O(np^{1/2})$ edges.
\end{theorem}
Before giving a more involved proof of this theorem, Coppersmith and Elkin point out a simple construction that nearly works, which can easily be converted to the following:
\begin{theorem} [\cite{CE06}] \label{thm:cedpforeach}
Any $n$-node directed weighted graph $G$ and $|P|=p$ demand pairs have a distance preserver with slack on $O(np^{1/2})$ edges.
\end{theorem}
\begin{proof}
Let $\ell$ be a parameter, and say that a demand pair $(s, t) \in P$ is ``short'' if it has a shortest path on $\le \ell$ edges, or ``long'' otherwise.
\begin{itemize}
\item For any short demand pair $(s, t) \in P$, add all edges of a shortest path $\pi(s, t)$ to the distance preserver (cost $O(p \ell)$).

\item To handle the long demand pairs $(s, t) \in P$, let $R$ be a random sample of nodes obtained by including each node independently with probability $\ell^{-1}$, and add in- and out- shortest path trees rooted at each $r \in R$ (expected cost $O(n^2 / \ell)$).
With constant probability or higher we sample a node $r \in R$ on a shortest path $\pi(s, t)$, and thus there is a shortest $s \leadsto t$ path included between the in- and out- shortest path trees rooted at $r$.
So each long demand pair is satisfied with constant probability or higher.
\end{itemize}
The proof now follows by setting $\ell := n/p^{1/2}$, giving total cost
\begin{align*}
|E(H)| = O\left( \ell p + n^2 / \ell \right) = O\left( n p^{1/2} \right). \tag*{\qedhere}
\end{align*}
\end{proof}

In fact, by Lemma \ref{lem:firstmain}, Theorem \ref{thm:cedpforeach} \emph{implies} Theorem \ref{thm:cedp}, and so this simpler proof suffices for Theorem \ref{thm:cedp}.
For another example along these lines, the following facts are proved in \cite{AB18}:

\begin{theorem} [\cite{AB18}] \label{thm:reachubs} ~
\begin{enumerate}
\item Any $|P|=p$ demand pairs in an $n$-node directed graph $G = (V, E)$ has a reachability preserver on $O((np)^{2/3} + n)$ edges.

\item When $P \subseteq S \times V$ for some subset of $|S|=s$ nodes, there is a distance preserver on $O((nps)^{1/2} + n)$ edges.
\end{enumerate}
\end{theorem}

The two parts of this theorem are proved separately in \cite{AB18}, each using somewhat involved arguments.
We show that the former actually follows from the latter, thus cutting the work in half.
\begin{proof} [Proof of Theorem \ref{thm:reachubs}.1, given Theorem \ref{thm:reachubs}.2]

When $p = O(n^{1/2})$, we trivially have $P \subseteq S \times V$ for some node subset of size $s \le p$.
Hence, by Theorem \ref{thm:reachubs}.2, there is a complete reachability preserver on $O((np^2)^{1/2} + n) = O(n)$ edges.
When $p = \Omega(n^{1/2})$, we construct a reachability preserver with slack as follows.
Like before, let $\ell$ be a parameter, and say that a demand pair $(s, t) \in P$ is ``short'' if its shortest path (or any canonical choice of $s \leadsto t$ path will work here) has length $\le \ell$, or ``long'' otherwise.
\begin{itemize}
\item To handle the short pairs $(s, t)$, add the $\le \ell$ edges of a path to the preserver (cost $O(p \ell)$).

\item To handle the long pairs $(s, t)$, randomly sample a set of nodes $R$ by including each node independently with probability $\ell^{-1}$.
Let $P_R$ denote the demand pairs $(s, t)$ whose shortest path intersects a node $r \in R$, and note that each long pair $(s, t)$ is in $P_R$ with at least constant probability.
We then split each such pair $(s, t) \in P_R$ into two pairs $(s, r), (r, t)$ and add two reachability preservers via Theorem \ref{thm:reachubs}.2, to handle all pairs of the form $(s, r)$ and then all pairs of the form $(r, t)$, for cost
$$O\left(\sqrt{|R||P_R|n} + n\right) = O\left(np^{1/2} / \ell^{1/2} + n\right).$$
\end{itemize}
The proof now follows by setting $\ell := n^{2/3} / p^{1/3}$, giving total cost
\begin{align*}
|E(H)| = O\left(p\ell + np^{1/2}/\ell^{1/2} + n\right) = O\left(n^{2/3} p^{2/3} + n\right). \tag*{\qedhere}
\end{align*}
\end{proof}

We next turn to pairwise spanners.
The following auxiliary lemma will be useful.
Let us say that a \emph{$d$-initialization} of a graph $G$ is a subgraph $H$ obtained by arbitrarily choosing $d$ edges incident to each node in $G$ and including them in $H$, or including all edges incident to a node of degree $\le d$ (this simplifying technique, which replaces the standard \emph{clustering step}, was first used in \cite{Knudsen14}).
\begin{lemma} [e.g.\ \cite{Knudsen14, Chechik13soda} and others]\label{lem:initnbhd}
If $H$ is a $d$-initialization of an undirected unweighted graph $G$, and there is a shortest path $\pi$ in $G$ that is missing $x$ edges in $H$, then there are $\Omega(xd)$ total nodes adjacent in $H$ to any node in $\pi$.
\end{lemma}
\begin{proof}
Note that any node $y$ is adjacent to at most three nodes in $\pi$, since otherwise there is a path of length $2$ (passing through $y$) between the first and last such node, which is shorter than the corresponding subpath in $\pi$.
Additionally, for each edge $(u, v) \in \pi \setminus H$, there must be $\ge d$ edges in $H$ incident to $u, v$ since we did not choose to add $(u, v)$ itself in the initialization.
Thus, we have:
\begin{align*}
\left|\{x \ \mid \ x \text{ adjacent to } \pi\}\right| &\ge \frac{\sum \limits_{(u, v) \in \pi \setminus H} \deg_H(u)}{3}\\
&= \Omega\left(xd\right). \tag*{\qedhere}
\end{align*}
\end{proof}

Using this, we now give some hitting-set-based pairwise spanner constructions.
We will first prove:

\begin{theorem} [\cite{KV15, AB16soda}] \label{thm:twospan}
Every set of $|P| = p$ demand pairs in an $n$-node graph $G$ has a $+2$ pairwise spanner on $O(np^{1/3})$ edges.
\end{theorem}
Kavitha and Varma \cite{KV15} implicitly proved a pairwise spanner with slack of this quality, while the complete version was subsequently proved in \cite{AB16soda} with a more involved argument.
The former proof is:
\begin{theorem} [\cite{KV15}]
Every set of $|P| = p$ demand pairs in an $n$-node graph $G$ has a $+2$ pairwise spanner with slack on $O(np^{1/3})$ edges.
\end{theorem}
\begin{proof}
Let $\ell, d$ be parameters, and like before, say that a demand pair $(s, t) \in P$ is ``short'' if its shortest path is currently missing $\le \ell$ edges in the spanner, or ``long'' otherwise.
Start the spanner as a $d$-initialization of $G$ (cost $O(nd)$).
Then:

\begin{itemize}
\item For the short pairs $(s, t)$, add the $\le \ell$ missing edges of a shortest path to the spanner (cost $O(p \ell)$).

\item To handle the long pairs $(s, t)$, randomly sample a set of nodes $R$ by including each node independently with probability $(\ell d)^{-1}$.
Add to the spanner a shortest path tree rooted at each $r \in R$ (cost $O(n^2/(\ell d))$).
By Lemma \ref{lem:initnbhd} there are $\Omega(\ell d)$ nodes adjacent to the shortest $s \leadsto t$ path, so with constant probability or higher, we sample a node $r \in R$ adjacent to a node $u$ on this shortest path.
In this event, we compute:
\begin{align*}
\dist_H(s, t) &\le \dist_H(s, r) + \dist_H(r, t) \tag*{triangle inequality}\\
&= \dist_G(s, r) + \dist_G(r, t) \tag*{shortest path tree at $r$}\\
&\le \dist_G(s, u) + \dist_G(u, t) + 2 \tag*{triangle inequality}\\
&= \dist_G(s, t) + 2 \tag*{$u$ on shortest $s \leadsto t$ path.}
\end{align*}
\end{itemize}

To complete the proof we then set $\ell := n/p^{2/3}$ and $d := p^{1/3}$, giving
\begin{align*}
|E(H)| = O\left(nd + p\ell + n^2/(\ell d)\right) = O\left(np^{1/3} \right). \tag*{\qedhere}
\end{align*}
\end{proof}

A similar story holds for the $+6$ pairwise spanner.
Kavitha \cite{Kavitha17} proved:
\begin{theorem} [\cite{Kavitha17}] \label{thm:sixspan}
Every set of $|P| = p$ demand pairs in an $n$-node graph $G$ has a $+6$ pairwise spanner on $O(np^{1/4})$ edges.
\end{theorem}

Kavitha also mentions a simpler proof that results in a pairwise spanner with slack on $O(np^{1/4})$ edges.
By our Lemma \ref{lem:firstmain}, in fact, this simpler proof implies Theorem \ref{thm:sixspan}.
The spanner with slack is constructed by reduction to the following key lemma in the area, which has been repeatedly rediscovered:
\begin{theorem} [\cite{BKMP10, Pettie09, ElkinUnpub, CGK13}] \label{thm:2subspan}
For every $n$-node undirected unweighted graph $G = (V, E)$ and set of demand pairs with the structure $P = S \times S$ for some $S \subseteq V, |S| = s$, there is a $+2$ pairwise spanner of $(G, P)$ on $O(ns^{1/2})$ edges.
\end{theorem}

We will not recap the proof of Theorem \ref{thm:2subspan} here.
Given this theorem, the spanner with slack is proved as follows:
\begin{theorem} [\cite{Kavitha17}]
Every set of $|P| = p$ demand pairs in an $n$-node graph $G$ has a $+6$ pairwise spanner with slack on $O(np^{1/4})$ edges.
\end{theorem}
\begin{proof}
Let $\ell, d$ be parameters, and start the spanner as a $d$-initialization of $G$ (cost $O(nd)$).
A demand pair $(s, t) \in P$ is ``short'' if the shortest $s \leadsto t$ path is missing $\le \ell$ edges in the spanner, or ``long'' otherwise.
\begin{itemize}
\item To handle the short pairs $(s, t)$, add the $\le \ell$ missing edges in its shortest path to the spanner (cost $O(p \ell)$).

\item To handle the long demand pairs $(s, t)$, there are two steps.
First, add the first and last $\ell$ missing edges of the shortest $s \leadsto t$ path to the spanner (cost $O(p \ell)$).
Then, randomly sample a set $R$ by including each node with probability $(\ell d)^{-1}$.
Using Theorem \ref{thm:2subspan}, add a $+2$ pairwise spanner on demand pairs $R \times R$; this costs
$$O\left(n \sqrt{|R|}\right) = O\left(n^{3/2} / \sqrt{\ell d}\right)$$
edges.
By Lemma \ref{lem:initnbhd} the added prefix and suffix of the shortest $s \leadsto t$ path each have $\Omega(\ell d)$ adjacent nodes.
Thus, with constant probability or higher, we sample $r_1, r_2 \in R$ such that $r_1$ is adjacent to $u_1$ in the added prefix and $r_2$ is adjacent to $u_2$ in the added suffix.
In this event we can compute:
\begin{align*}
\dist_H(s, t) &\le \dist_H(s, r_1) + \dist_H(r_1, r_2) + \dist_H(r_2, t) \tag*{triangle inequality}\\
&\le \dist_H(s, r_1) + (\dist_G(r_1, r_2) + 2) + \dist_H(r_2, t) \tag*{$R \times R$ $+2$ pairwise spanner}\\
&\le (\dist_H(s, u_1) + 1) + (\dist_G(r_1, r_2) + 2) + (\dist_H(u_2, t) + 1) \tag*{triange inequality}\\
&= \dist_G(s, u_1) + \dist_G(r_1, r_2) + \dist_G(u_2, t) + 4 \tag*{added prefix/suffix}\\
&\le \dist_G(s, u_1) + \dist_G(u_1, u_2) + \dist_G(u_2, t) + 6 \tag*{triangle inequality}\\
&= \dist_G(s, t) + 6 \tag*{$u_1, u_2$ on $s \leadsto t$ shortest path.}
\end{align*} 
\end{itemize}
To complete the proof we set $\ell := n/p^{3/4}$ and $d := p^{1/4}$, giving
\begin{align*}
|E(H)| = O\left(nd + p\ell + n^{3/2}/\sqrt{\ell d} \right) = O\left( np^{1/4} \right). \tag*{\qedhere}
\end{align*}
\end{proof}

Finally, we discuss the $+4$ pairwise spanner.
Kavitha \cite{Kavitha17} proved a $+4$ pairwise spanner on $\Oish(np^{2/7})$ edges, which can easily be turned into a $+4$ pairwise spanner with slack on $O(np^{2/7})$ edges.
By Lemma \ref{lem:firstmain}, in fact this implies a complete $+4$ pairwise spanner on $O(np^{2/7})$ edges, thus shaving the $\log$ factors from the original result in \cite{Kavitha17}.
Kavitha's proof is as follows:

\begin{theorem} [\cite{Kavitha17}] \label{thm:fourspan}
Every set of $|P| = p$ demand pairs in an $n$-node graph has a $+4$ pairwise spanner with slack on $O(np^{2/7})$ edges.
\end{theorem}
\begin{proof}
Let $\ell, d$ be parameters, and let the spanner be a $d$-initialization of $G$ (cost $O(nd)$).
This time there are three cases: a demand pair $(s, t)$ is ``short'' if its shortest path is missing $\le \ell$ edges, it is ``medium'' if its shortest path is missing $>\ell$ and $\le n/d^2$ edges, or it is ``long'' otherwise.
\begin{itemize}
\item To handle the short pairs $(s, t)$, add the $\le \ell$ missing edges of the shortest path to the spanner (cost $O(p \ell)$).

\item To handle the long pairs $(s, t)$, randomly sample a set of nodes $R_1$ by including each node independently with probability $d/n$, and add the edges of a BFS tree rooted at each $r \in R_1$ to the spanner (cost $O(nd)$).
By Lemma \ref{lem:initnbhd} there are $\Omega(n/d)$ nodes adjacent to the shortest $s \leadsto t$ path, so with constant probability or higher we sample a node $r \in R_1$ adjacent to a node $u$ on this path.
In this event, we compute:
\begin{align*}
\dist_H(s, t) &\le \dist_H(s, r) + \dist_H(r, t) \tag*{triangle inequality}\\
&= \dist_G(s, r) + \dist_G(r, t) \tag*{shortest path tree}\\
&\le \dist_G(s, u) + \dist_G(u, t) + 2 \tag*{triangle inequality}\\
&= \dist_G(s, t) + 2 \tag*{$u$ on a shortest $s \leadsto t$ path.}
\end{align*}

\item There are two steps to handle the medium pairs $(s, t)$.
First, add the first and last $\ell$ missing edges in the shortest path to a spanner (cost $O(p \ell)$).
Then, randomly sample a set of nodes $R_2$ by including each node independently with probability $(\ell d)^{-1}$.
For each pair of nodes $r, r' \in R_2$, check to see if there exist nodes $u, u'$ adjacent to $r, r'$ (respectively) in the current spanner $H$ such that the shortest $u \leadsto u'$ path is missing $\le n/d^2$ edges.
If so, then choose nodes $u, u'$ with this property minimizing $\dist_G(u, u')$, and add all missing edges in the shortest $u \leadsto u'$ path to the spanner.
If no such nodes $u, u'$ exist, then do nothing for this pair $r, r' \in R$.
This step costs
$$O\left( |R_2|^2 \cdot \frac{n}{d^2} \right) = O\left( \frac{n^2}{\ell^2 d^2} \cdot \frac{n}{d^2} \right) = O\left( \frac{n^3}{\ell^2 d^4} \right)$$
edges.
For a medium demand pair $(s, t)$, by Lemma \ref{lem:initnbhd} there are $\Omega(\ell d)$ nodes adjacent to the added prefix and suffix, so with constant probability or higher we sample nodes $r, r' \in R_2$ adjacent to nodes $x, x'$ on the added prefix, suffix (respectively).
In this event, note that there are $\le n/d^2$ missing edges on the shortest $x \leadsto x'$ path, since $x, x'$ are on the $s \leadsto t$ shortest path and $(s, t)$ is a medium pair.\footnote{A technical detail here is that this step requires that shortest paths are chosen \emph{consistently}, i.e.\ the canonical shortest $x \leadsto x'$ path is a subpath of the canonical shortest $s \leadsto t$ path.}
Thus, when $r, r' \in R_2$ are considered in the construction, we will indeed add a new shortest path to the spanner (as opposed to the case where we do nothing).
Letting $u, u'$ be the endpoints of this added shortest path, we compute:
\begin{align*}
\dist_H(s, t) &= \dist_H(s, x) + \dist_H(x, x') + \dist_H(x', t) \tag*{$x, x'$ on shortest $s \leadsto t$ path}\\
&= \dist_G(s, x) + \dist_H(x, x') + \dist_G(x', t) \tag*{$x, x'$ on added prefix, suffix}\\
&\le \dist_G(s, x) + \left(\dist_H(u, u') + 4\right) + \dist_G(x', t) \tag*{triangle inequality}\\
&= \dist_G(s, x) + \dist_G(u, u') + \dist_G(x', t) + 4 \tag*{shortest $u \leadsto u'$ path added}\\
&\le \dist_G(s, x) + \dist_G(x, x') + \dist_G(x', t) + 4 \tag*{$\dist_G(u, u')$ minimal}\\
&= \dist_G(s, t) + 4 \tag*{$x, x'$ on shortest $s \leadsto t$ path.}
\end{align*}
\end{itemize}
We then complete the proof by setting $\ell := n/p^{5/7}$ and $d := p^{2/7}$, giving
\begin{align*}
|E(H)| = O\left(nd + p\ell + n^3/(\ell^2 d^4) \right) = O\left(np^{2/7}\right). \tag*{\qedhere}
\end{align*}
\end{proof}

This completes the proof(s) of Theorem \ref{thm:sparsifiers}.

\section*{Acknowledgments}

I am grateful to Stephen Kobourov, Reyan Ahmed, Richard Spence, Richard Peng, and two anonymous reviewers for comments and questions that have improved the quality of this writeup.

\bibliographystyle{acm}
\bibliography{references}

\end{document}